\documentclass[12pt]{article}
\usepackage{mathtools}
\usepackage{amsmath,amsfonts,amssymb,amsthm,hyperref,url}
\newcommand{\ket}[1]{\left|#1\right\rangle}

\newtheorem{thm}{Theorem}
\newtheorem{conj}{Conjecture}
\newtheorem{lemma}{Lemma}
\title{\LARGE Negativity Bounds for Weyl-Heisenberg Quasiprobability Representations}
\author{\normalsize John B. DeBrota and Christopher A. Fuchs\\
\normalsize University of Massachusetts Boston}
\date{}

\begin{document}

\maketitle
\begin{abstract}
    The appearance of negative terms in quasiprobability representations of quantum theory is known to be inevitable, and, due to its equivalence with the onset of contextuality, of central interest in quantum computation and information. Until recently, however, nothing has been known about how much negativity is necessary in a quasiprobability representation. Zhu \cite{Zhu2016} proved that the upper and lower bounds with respect to one type of negativity measure are saturated by quasiprobability representations which are in one-to-one correspondence with the elusive symmetric informationally complete quantum measurements (SICs). We define a family of negativity measures which includes Zhu's as a
    special case and consider another member of the family which we call ``sum negativity.'' We prove a sufficient condition for local maxima in sum negativity and find exact global maxima in dimensions $3$ and $4$. Notably, we find that Zhu's result on the SICs does not generally extend to sum negativity, although the analogous result does hold in dimension $4$. Finally, the Hoggar lines in dimension $8$ make an appearance in a conjecture on sum negativity.
\end{abstract}
\section{Introduction}
The mathematical machinery of quantum theory has persisted without substantial modification for nearly a century, but we are still waiting for a compelling set of physical principles upon which to hang the theory's predictions. It is the hope of some quantum foundations researchers that looking at standard quantum theory re-represented in an appropriate fashion will help suggest these
principles \cite{weinberg17}.

Probabilities are central objects in canonical quantum theory; at the end of a quantum mechanical calculation, we are left with a probability distribution or a simple consequence of one. It is tempting, therefore, to behave as though quantum theory \textit{gives} us a probability distribution---\textit{the} probabilities for a set of outcomes. Careful consideration reveals, however, that \textit{given a quantum state}, quantum theory allows us to calculate a probability distribution.
But where does the quantum state come from in the first place? The apparatus of quantum theory is unable to say. In practice, an experimenter eventually settles on a quantum state for her preparation procedure after a suite of tests and calibrations, and, ultimately, statistical methods go into the state determination itself. That is, we have probabilities at the beginning and probabilities at the end  with the formal apparatus of quantum theory gluing it all together. It is possible that what goes on
in between these ends stands alone, but such a circumstance is far from guaranteed. The proper understanding of quantum theory may depend upon conceiving of probability theory in the proper way.

In fact, revisiting and deciding on the proper understanding of probability theory is the starting point of QBism \cite{HerosHandbook,LordVoldemort,RMP}. QBists take a strict personalist Bayesian \cite{DeFinetti90,Lad96,Bernardo94} stance on probability theory. A probability is a valuation an agent places on his or her degree of belief in a possible outcome, nothing more and nothing less. As a consequence, probabilities are not empirically determinable quantities because they do not
independently exist outside of an agent's mind.

As an example, consider repeated flips of a coin. A frequentist conception of probability asserts that the probability of heads for the coin is the long-run ratio of number of heads to number of flips. The Bayesian first points out that to regard the coin flipped at different times as an ``equivalent'' or ``exchangeable'' process amounts to a belief the experimenter has about the situation---perhaps nothing is wrong with this belief, but she should be cognizant of its influence on the
conclusions of the experiment. Secondly, the Bayesian asks just how many times the experimenter plans to actually flip the coin before she decides that the relative frequency is the probability of heads for that coin. If she flips it a finite number of times, by her own admission, any frequency is technically possible (although she believes some are unlikely). The usual answer to such a question is that she plans to flip it until the deviations in the ratio with further experimentation
are small enough
so as to be negligible. In other words, in order to \textit{define} the likelihood of an outcome, she asserts that the ratio obtained by a finite series of experiments is \textit{likely} to be close to the ``true probability''. The circularity of such an argument should be evident and worrying to anyone espousing the frequentist paradigm. Properly understood, then, probabilities are single-case; no probability is meaningfully right or wrong by any \textit{external} criterion. The knee-jerk
reaction to this statement is to exasperatedly throw up one's hands and exclaim that Bayesianism is just the claim that probability theory is useless! Nothing could be further from the truth: Although a probability is not subject to objective external validation or invalidation, it commands the same sway over our lives that it would if it were; if an agent wishes to avoid sure loss---to avoid being demonstrably stupid---she must take steps to ensure that she never assigns
probabilities which mutually contradict each other. When a theory tells us what we ``should'' do or ``strive for'', it is a \textit{normative} theory. In other words, the personalist Bayesian view is that probability theory is a normative theory. In this community, compatible probabilities are called \textit{coherent}. Remarkably, nearly all of the standard rules of probability theory are consequences of coherence \cite{Lad96}.

One consequence that we will reference soon
is the Law of Total Probability (LTP):
\begin{equation}\label{ltp}
    q(j)=\sum_i p(i)r(j|i).
\end{equation}
The LTP describes a scenario where two actions are taken, one after another. $p(i)$ represents the probability we ascribe to getting outcome $i$ from the first action, $r(j|i)$ is the probability we ascribe to getting outcome $j$ from the second action conditioned on outcome $i$ for the first action, and $q(j)$ is the probability we ascribe to outcome $j$ for the second action, not conditioned on anything other than the operational procedure we have laid out. The commitment to coherence alone (and independent from any possible nature of reality) requires that our probabilities assigned at any given moment should hold together in
accordance with \eqref{ltp}.

As we explain below, it is possible to represent any quantum state as a single probability distribution over the possible outcomes of an appropriately chosen measurement. If we take this fact seriously, a quantum state is conceptually nothing more than a probability distribution. In QBism, all of the personalist Bayesian properties of probability theory carry over to quantum states; that is, quantum states, like probabilities, are valuations of belief for future
experiences. However, application of the rules of quantum theory reveal that not all of the probability distributions in the probability simplex correspond to a valid quantum state. Just as in probability theory where an agent strives to be consistent with herself in her probability assignments, an agent should not ascribe a probability distribution she knows to be in conflict with the quantum mechanical formalism. In this way we arrive at an understanding that quantum theory is an empirically-motivated normative addition to probability theory.

If the functional form of the additions to probability calculus are cumbersome, then there may be no reason to adopt it for everyday use---furthermore, it may not shed any light on the ``nature of reality.'' What would constitute a nice looking addition to probability theory? One possibility would be if the normative rules of quantum theory could be made to mirror those of probability theory in a suggestive way. It turns out that just this sort of situation can be made to occur.

An \textit{informationally complete quantum measurement (IC-POVM)} for a Hilbert space $\mathcal{H}_d$ is a set of at least $d^2$ positive semi-definite operators $E_i$ which span $\mathcal{L}(\mathcal{H}_d)$, the vector space of linear operators on $\mathcal{H}_d$, and satisfy $\sum_{i} E_i=\mathbb{I}$. When such a measurement consists of exactly $d^2$ elements, density matrices $\rho$ and the Born rule probabilities $p(i)=\text{Tr}(\rho E_i)$ are in bijective correspondence because the $E_i$ form a basis for $\mathcal{L}(\mathcal{H}_d)$. Such minimal IC-POVMs are known to exist in all dimensions
\cite{Caves02b}. Which
one we choose for a representation, however, stands a chance of revealing or obscuring the properties which probability distributions equivalent to quantum states must have. Very often in mathematics and physics, a hard problem becomes easy when we choose the right basis. For example, the Eddington--Finkelstein coordinates revealed that the event horizon of a non-rotating black hole is not a physical singularity. Is there a best or particularly nice IC-POVM which will reveal hidden properties of quantum theory?

It is not possible for an IC-POVM to be an orthonormal basis \cite{Spekkens08,Ferrie08,Appleby07}, so perhaps our first hope is ruled out. If the elements of an IC-POVM cannot be orthogonal, can they at least be equiangular? It turns out that they can \cite{appleby07b}. That is, we can find a set of $E_i$ which satisfy $\text{Tr}(E_iE_j)=c$ for $i\neq j$. Can such a POVM consist of only rank-one matrices? Remarkably for such a simply-stated question, it is not generally known. A set of $d^2$ rank-one matrices $\Pi_i$ such that
\begin{equation}
    \text{Tr}\left(\Pi_i\Pi_j\right)=\frac{d\delta_{ij}+1}{d+1}
\end{equation}
defines an IC-POVM $E_i=\frac{1}{d}\Pi_i$ called a \textit{Symmetric IC-POVM (SIC)} \cite{Zauner99,Caves02d,Renes04}. High-precision numerical SICs have been found
in all dimensions $2$--$151$ \cite{sicquestion,scott17} and in a few sporadic higher dimensions. In many cases, exact SICs have been constructed among these dimensions as well \cite{newexact}. All indications are that SICs exist in all dimensions, but the proof continues to evade us.

If a SIC exists in dimension $d$, it is possible to rewrite the Born rule in a uniquely simple form analogous to the LTP, an equation called the \textit{urgleichung} \cite{RMP,LordVoldemort}, German for ``primal equation'' in reference \cite{RMP}:
       \begin{equation}\label{primal}
              q(j)=\sum_i \left[(d+1)p(i)-\frac{1}{d}\right]r(j|i),
      \end{equation}
      where $q(j)$ are the probabilities for obtaining outcome $j$ of a general quantum measurement, $p(i)$ is the probability an agent ascribes to obtaining outcome $i$ in the imagined scenario where a SIC measurement is performed on the system instead, and $r(j|i)$ are the probabilities for obtaining the equivalent outcome $j$ conditional on obtaining outcome $i$ that the agent ascribes in the imagined scenario (See \cite{HerosHandbook} for a detailed exposition). It is
      essential to recognize the operational difference between the urgleichung and the LTP: the urgleichung describes a scenario where the first measurement is not \textit{actually} made---just imagined. If we actually planned to implement the first measurement, our probabilities must hold together according to the familiar LTP. We take equation \eqref{primal} very seriously. In fact, it motivated the
most recent development in QBism---a reconstruction of quantum theory featuring a generalization of the urgleichung as the key assumption \cite{qplex}. See references \cite{Khrennikov15} and \cite{Khrennikov16} for critical review and discussion of the urgleichung in QBism as well as comparison to other contexts.

Probability theory itself has no tether to physical reality---rather it is a tool that anyone, \textit{anywhere}, can use to manage their expectations for further experiences. Those expectations will certainly be influenced by deeply-held convictions that the agent has about the nature of reality around them, but the way those probabilities must hang together if the agent is to be coherent is unaffected. Quantum mechanics, on the other hand, although evidently a normative theory like probability
theory, \textit{is} tethered to physical reality. A reformulation of quantum theory which brings to the forefront this normative structure allows us to examine the threads of this tether without confusing the subjective and the objective. This way, we may hopefully more readily determine the aspects of reality which forced quantum theory to be the way that it is. The LTP is a direct consequence of coherence in one's probability assignments. Is the urgleichung, which presents as an almost trivial modification of the LTP, on the right track for expressing the conditions for a kind of quantum coherence that an agent should strive for by virtue of being in our universe? We would like to accumulate as much evidence as possible that it is. Often one can gather more evidence for a sentiment simply by looking where it's least expected. In this case, Huangjun Zhu recently demonstrated additional evidence for this line of reasoning by instead departing from probability theory \cite{Zhu2016}.

If we formally relax the positivity condition for minimal IC-POVMs (keeping the fact that they sum to the identity), we are dealing with the larger space of Hermitian operator bases. Denoting such an operator $F_i$, if we also keep the form of the Born rule for a quantum state $\rho$, we obtain a set of real numbers $\mathfrak{p}(i)=\text{Tr}(\rho F_i)$, some of which may be negative, such that $\sum_i\mathfrak{p}(i)=1$. This set of numbers is referred to as a \textit{quasiprobability vector} (we will always denote
quasiprobability vectors with $\mathfrak{fraktur}$ script) and a quasiprobability vector obtained in this way from a quantum state is called a \textit{quasiprobability representation} of the state.

Quantum opticians have benefited from the ease of plotting quasiprobability distributions over phase space \cite{Schleich01}, so there is some utility in their use, but what is a quasiprobability? If a probability is a valuation of belief, what meaning can we attach to a quasiprobability? There does not appear to be a simple meaning---some attempts at attaching operational
substance to quasiprobabilities have been made, for example, see references \cite{Prugovecki73,Feynman87,Khrennikov97,Khrennikov09,Burgin10}, but if these solutions get us no closer to understanding \textit{why} quantum theory is the normative probability calculus an agent of our universe should use to productively navigate, then they amount to duct tape over a structural weakness. However unsatisfying a quasiprobability may be on principled grounds, it turns out that permitting them for the time being gets us something desirable in return: the Born rule obtains an even closer functional analogy to the LTP.

Consider a Hermitian operator basis $\{F_i\}$ such that $\sum_i F_i=1$ and a dual basis\footnote{A dual basis is one for which the bases considered together are biorthogonal, $\text{Tr}F_iQ_j=\delta_{ij}$.} $\{Q_j\}$ constrained to satisfy $\text{Tr}(Q_j)=1$. From a state $\rho$ and an arbitrary POVM $\{G_j\}$ we form the quasiprobabilities $\mathfrak{p}(i)=\text{Tr}(\rho F_i)$ and the conditional quasiprobabilities $\mathfrak{r}(j|i)=\text{Tr}(Q_iG_j)$. Using the identity
\begin{equation}\label{innerprodidentity}
    \sum_j\text{Tr}(BF_j)\text{Tr}(Q_j C)=\text{Tr}(BC),
\end{equation}
we can rewrite the Born rule as
            \begin{equation}\label{qprimal}
                q(j)=\text{Tr}(\rho G_j)=\sum_i\mathfrak{p}(i)\mathfrak{r}(j|i).
            \end{equation}
            Like the urgleichung, \eqref{qprimal} looks very much like the LTP. In fact, \eqref{qprimal} is functionally \textit{identical} to the LTP. The difference is that this equation is written in terms of quasiprobabilities instead of probabilities. Negativity must pop up somewhere, for it is known that negativity \textit{must}\footnote{That is, it is impossible to represent quantum theory in a way which eliminates the appearance of negativity in \textit{both} $\mathfrak{p}$
            and $\mathfrak{r}$ in \eqref{qprimal} for all quantum states and POVMs. It is possible, in general, to eliminate the negativity appearing in one or the other.} appear in quasiprobability representations of
quantum theory \cite{Ferrie08,Ferrie10}.

Another advantage when we are not burdened with positivity is that we may choose the $F_i$ to form an orthogonal basis for operator space. If a basis is orthogonal, it is proportional to its dual basis and called \textit{self-dual}. In this case, the sum constraint on the basis automatically fixes $\text{Tr}(Q_j)=1$ and the constant of proportionality $F_j=\frac{1}{d}Q_j$. Here we depart from Zhu's terminology and refer to a quasiprobability representation obtained from this sort of self-dual basis as a \textit{Q-rep}. Q-reps
account for most of the quasiprobability representations considered in the
literature \cite{Wootters87,Wootters04,Ferrie09,Ferrie11}. Importantly, however, Q-reps do not account for all finite dimensional quasiprobability representations; any nonorthogonal basis provides an example outside of this set. We can identify a Q-rep with the dual basis, $\{Q_i\}$, which defines it. We
will reserve $\mathfrak{q}$ for quasiprobability representations of states with respect to a Q-rep $\{Q_j\}$ (when there is no index, we are referring to the full vector).

Mathematically speaking, a probability is a type of quasiprobability; that is, the kind without any negative entries. If we believe that the urgleichung differing from
the LTP captures some aspect of the essential difference between quantum and classical, then some of this essential difference is also contained in the difference between the LTP and \eqref{qprimal}, that is, in the appearance of negativity in quasiprobability representations of quantum theory. Indeed, Spekkens showed that the presence of negative elements in quasiprobability representations of quantum theory and the impossibility of noncontextual hidden variable models are equivalent notions of
nonclassicality \cite{Spekkens08}. Therefore, insofar as we think contextuality is an important ingredient in the quantum--classical distinction, we should be interested in negativity as well. Investigating the negativity in Q-reps also seems to be a promising approach to identifying exactly what advantages quantum computation affords us over classical computation, for example, Veitch \textit{et al.} showed that negativity is a resource for quantum computation \cite{Emerson12} and Howard
\textit{et al.} recently showed that contextuality enables universal quantum computation via `magic state' distillation \cite{Emerson14}. Additionally, and, as we will see, of particular note for this paper, Pashayan \textit{et al.} have shown that a value related to the sum of the negative entries in a quasiprobability representation may be thought of as a measure that bounds the efficiency of a classical estimation of probabilities \cite{Bartlett15}. In light of these facts, Zhu's recent result is especially exciting. In his paper, Zhu developed a natural measure of negativity for Q-reps and established strict upper and lower bounds for this measure in one-to-one correspondence with SICs in each dimension.
SICs are related to the bounds of this measure of negativity and the appearance of negativity in quasiprobability representations of quantum theory seems to contain some hints toward what ``quantum'' really means. Are there more hints to be found? Motivated to answer this question, we investigate how robust Zhu's result is to modifications in the negativity measure.

In Section \ref{negandsumneg} we motivate and define a general negativity measure for quasiprobability representations which includes Zhu's measure and the measure of principle interest in this paper, called sum negativity, as special cases and state Zhu's theorem that the bounds on his measure of negativity are achieved by Q-reps in one-to-one correspondence with SICs. In Section \ref{sumnegsic} we address the sum negativity for the SIC Q-reps in the first few dimensions. In Section \ref{whnqpr} we argue that Weyl-Heisenberg covariant Q-reps are a natural subset to
consider while looking for counterexamples to Zhu's theorem and establish the explicit conditions for a Weyl-Heisenberg covariant Q-rep in dimension $3$. Section \ref{sumnegd3} contains the main results: we explicitly demonstrate that Zhu's theorem does not generally extend to sum negativity in either bound, we prove a general sufficiency theorem for a Q-rep being a local maximum for sum negativity, use this theorem to prove the exact upper bound for sum
negativity among Q-reps in dimension $3$, and state a conjecture regarding the lower bound among Weyl-Heisenberg Q-reps. In Section \ref{further} we again apply our theorem to prove that, although not generally the case, one of the SIC Q-reps achieves the exact upper bound for sum negativity in dimension $4$. We also briefly discuss the sum negativity for one of the Hoggar SIC Q-reps in dimension $8$. In Section \ref{discussion} we discuss further questions and directions.

\section{Negativity and Sum Negativity}\label{negandsumneg}

Setting aside Q-rep vectors for a moment, we start by proposing a family of negativity measures for general quasiprobability vectors. Qualitatively speaking, we want a measure of the ``amount'' of negativity that appears in a vector with entries which sum to $1$. A few candidates immediately stand out as especially natural measures: Perhaps a measure proportional to the sum of the negative elements or to the most negative element appearing in the quasiprobability vector---indeed, as we will see, the latter choice is taken by Zhu in \cite{Zhu2016}. We want a family of measures of negativity to meaningfully capture the deviation or ``distance'' from quasiprobabilities which have no negative elements. So we are faced with the task of measuring something we might understand as a distance for elements in a finite dimensional vector space (constrained by the normalization condition, of course). From this vantage point, the $L^p$-norms offer a very compelling family of generalized distance measures which we might utilize. Recall the definition of the $L^p$-norm of a vector $\mathbf{x}$,
\begin{equation}
    |\!|\mathbf{x}|\!|_p:=\left(\sum_{i=1}^n|x(i)|^p\right)^{\!1/p},
\end{equation}
and the limiting expression for $p\to\infty$,
\begin{equation}
    |\!|\mathbf{x}|\!|_\infty=\max{\{|x_1|,|x_2|,\ldots,|x_n|\}}.
\end{equation}
We might hope that the $L^p$-norm of a quasiprobability vector will itself be a useful quantity which is immediately related to the negativity present in the vector. If $p=1$, for example, we can see that if there are any negative elements in the quasiprobability vector $\mathfrak{p}$ then $|\!|\mathfrak{p}|\!|_1>1$, so perhaps we could associate the amount by which the $L^p$-norm is larger than $1$ with the negativity of $\mathfrak{p}$. However, $|\!|\mathfrak{p}|\!|_p\leq|\!|\mathfrak{p}|\!|_q$ when $p>q$ and, in fact, $|\!|\mathfrak{p}|\!|_p$ need not be greater than $1$ even when negativity is present. So what we really want is to measure the deviation of only the negative part of the quasiprobability vector from the zero vector, disregarding all positive entries entirely. With this in mind, define the \textit{negative part} of a quasiprobability vector $\mathfrak{p}$ to be the vector
\begin{equation}\label{negpart}
        \mathfrak{p}^{(-)}(j):=\frac{|\mathfrak{p}(j)|-\mathfrak{p}(j)}{2}
    \end{equation}
    which replaces the positive elements of $\mathfrak{p}$ with zero and the negative elements with their absolute value. This definition makes it easy to isolate properties of the negative elements of a quasiprobability vector. Now that we have done away with the positive entries in our quasiprobability vector, we define the \textit{$N^p$ negativity} of a quasiprobability vector $\mathfrak{p}$ to be the $L^p$-norm of the negative part of a quasiprobability vector:
\begin{equation}\label{npneg}
    N^p(\mathfrak{p}):=\big|\!\big|\mathfrak{p}^{(-)}\big|\!\big|_p.
\end{equation}
We will refer to the special cases $N^1$ and $N^\infty$, which we see are equivalent to the two natural candidates proposed above, as the \textit{sum negativity}\footnote{Veitch \textit{et al.} use the term sum negativity specifically for the sum of the negative elements of the discrete Wigner function for a quantum state \cite{Emerson14b} whereas we will be considering the equivalent notion with respect to any Q-rep.} and the \textit{ceiling negativity} respectively.

Zhu defines the negativity of a quantum state $\rho$ with respect to a Q-rep $\{Q_j\}$ to be $d$ times the magnitude of most negative element appearing in the quasiprobability representation $\mathfrak{q}(j)=\text{Tr}(\rho F_j)$. In our framework, this corresponds to $d$ times $N^\infty(\mathfrak{q})$. He then defines the negativity of a Q-rep itself to be the maximum of this value over all of quantum state space. In our framework, this corresponds to $d$ times $\max_\rho
N^\infty(\mathfrak{q})$. Thus, in the general case, we define the \textit{$N^p$ negativity\footnote{We have chosen to omit multiplication by $d$ in our negativity definitions so that the negativity can more immediately be associated with the negative values in a quasiprobability vector. As the dual basis is calculationally easier to work with, a downside of our convention is that factors of $1/d$ crop up more frequently.} of a quantum state $\rho$ with respect to a Q-rep} $\{Q_j\}$ to be
\begin{equation}\label{neg_rho_wrt_Q}
    N^p(\rho,\{Q_j\}):=N^p(\mathfrak{q}),
\end{equation}
and the \textit{$N^p$ negativity of a Q-rep} $\{Q_j\}$ to be
\begin{equation}\label{neg_Q}
    N^p(\{Q_j\}):=\max_\rho N^p(\rho,\{Q_j\}).
\end{equation}
Since $L^p$-norms are convex and nondecreasing for $p \geq 1$ on the positive reals \cite{bhatia97} and the negative part map \eqref{negpart} is a convex function, the composition \eqref{neg_rho_wrt_Q} is also convex \cite{Bertsekas03}. This means the maximum occurs on the boundary of the domain so we may take the maximum in \eqref{neg_Q} to be over pure states in these cases.

As it will be useful later, note that the $L^2$ norm of a Q-rep vector corresponding to a pure state always equals $\sqrt{1/d}$. To see this, let  $B=C=\rho$ and $\text{Tr}(\rho^2)=1$ in \eqref{innerprodidentity}. This is another reason why the $L^p$ norms are not a good choice for a family of negativity measures.

Zhu notes that the ceiling negativity of a Q-rep $\{Q_j\}$ takes the simple form
\begin{equation}\label{negativity}
    N^\infty(\{Q_j\})=\frac{1}{d}\left|\min_j\lambda_{\text{min}}(Q_j)\right|,
\end{equation}
where $\lambda_\text{min}(Q_j)$ is the minimal eigenvalue of $Q_j$.

If a SIC, denoted $\{\Pi_j\}$, exists in dimension $d$, we may construct two Q-reps $\{Q^+_j\}$ and $\{Q^-_j\}$ called the \textit{SIC Q-reps}:
\begin{equation}\label{sicframes}
    Q_j^\pm=\mp \left(\sqrt{d+1}\right)\Pi_j+\frac{1}{d}\left(1\pm\sqrt{d+1}\right)\mathbb{I}
\end{equation}
which have ceiling negativities
\begin{equation}
   N^+ \equiv N^\infty(\{Q^+_j\})= \frac{(d-1)\sqrt{d+1}-1}{d^2}, \quad N^- \equiv N^\infty(\{Q^-_j\})= \frac{\sqrt{d+1}-1}{d^2}.
\end{equation}
We introduce these Q-reps because of the following theorem.

\begin{thm}[Zhu]\label{zhu}
    Every Q-rep $\{Q_j\}$ in dimension $d$ satisfies $N^-\leq N^\infty(\{Q_j\})\leq N^+$. The lower bound is saturated if and only if $\{Q_j\}$ has the form $\{Q^-_j\}$ where $\Pi_j$ is a SIC. If $\{Q_j\}$ is group covariant, then the upper bound is saturated if and only if $\{Q_j\}$ has the form $\{Q^+_j\}$.
\end{thm}

    Zhu's theorem identifies SICs as centrally important to the study of Q-reps and more broadly for quantum theory because the Q-reps which achieve both bounds on ceiling negativity in any dimension are related to SICs by a simple affine transformation. Is the ceiling negativity unique in this way? If so, it would be interesting to understand why. If not, where does it fail?

    We will address this question for the sum negativity. Like ceiling negativity, there is a more manageable expression for the sum negativity of a Q-rep. The following argument is due to Appleby and Zhu.
    \begin{lemma}\label{sumnegmethod}
        An equivalent form of the sum negativity of a Q-rep $\{Q_j\}$ is
\begin{equation}\label{sumneg2}
    N^1(\{Q_j\})=-\frac{1}{d}\min\left\{\lambda_{\{1\}},\lambda_{\{2\}},\ldots,\lambda_{\{2^{d^2}-1\}}\right\},
\end{equation}
where $\lambda_{\{i\}}$ is the minimal eigenvalue of the $i$th partial sum matrix of the $\{Q_j\}$ matrices.
    \end{lemma}
    \begin{proof}
        For a quasiprobability representation $\mathfrak{q}$ of a state $\rho$ with respect to Q-rep $\{Q_j\}$, if $\mathcal{S}:=\{i|\text{Tr}(\rho Q_i)<0\}$, then
        \begin{equation}\label{equivsumneg}
            N^1(\mathfrak{q})=\left|\frac{1}{d}\sum_{i\in\mathcal{S}}\text{Tr}(\rho Q_i)\right|=\left|\frac{1}{d}\text{Tr}\left(\rho \sum_{i\in\mathcal{S}}Q_i\right)\right|.
        \end{equation}
        As in \eqref{negativity}, the minimal value of the expression $\text{Tr}(\rho F)$ over quantum state space is the minimal eigenvalue of $F$ and the state $\rho$ which minimizes the expression is the corresponding eigenvector. Thus, for a fixed subset $\mathcal{G}$ of the $\{Q_j\}$ matrices, the state which minimizes $\text{Tr}(\rho\sum_{i\in\mathcal{G}}Q_i)$ is the minimal eigenvector of the matrix $\sum_{i\in\mathcal{G}}Q_i$. Via the definition above, for each $\rho$ there is a subset $\mathcal{S}$; in particular, there is a subset $\mathcal{S}'$ for a state whose quasiprobability representation has the sum negativity value $N^1(\{Q_j\})$ and, furthermore, the magnitude of the minimal eigenvalue of the matrix $\sum_{i\in\mathcal{S}'}Q_i$ is equal to $d$ times $N^1(\{Q_j\})$. Thus determining the sum negativity is equivalent to looking for the minimal eigenvalue over all partial sum matrices of $\{Q_j\}$. There are $2^{d^2}-1$ entries to minimize over because we ignore the partial sum corresponding to the empty subset.
\end{proof}
Does Theorem \ref{zhu} extend to sum negativity? In Section \ref{sumnegd3}, we will demonstrate that it generally does not with explicit counterexamples in the first nontrivial dimension, $d=3$.

\section{Sum Negativities of $\{Q_j^+\}$ and $\{Q_j^-\}$}\label{sumnegsic}
Sum negativity is notably harder to work with than ceiling negativity, both analytically and numerically. As such, using the method described in Lemma \ref{sumnegmethod}, analytic results for the sum negativity of $\{Q_j^+\}$ and $\{Q_j^-\}$ have only so far been obtained for dimensions 2, 3, and 4. Numerically exact results have also been obtained for $d=5$.

For $d=2$, sum negativity and ceiling negativity are equivalent measures because a $d=2$ Q-rep vector cannot contain more than one negative element. This property can be proven easily with equation \eqref{innerprodidentity} and the fact quasiprobabilities are normalized. In addition, the ceiling negativities for $\{Q^+\}$ and $\{Q^-\}$ are equivalent. Thus,
\begin{equation}
N^1(\{Q_j^\pm\})=N^\infty(\{Q_j^\pm\})=\frac{\sqrt{3}-1}{4}.
\end{equation}
In fact, all $N^p$ negativities for the SIC Q-reps are equal in dimension $2$. This is a reflection of the fact that all Q-reps are equivalent to the Wootters discrete Wigner function in this dimension \cite{Zhu2016}.

There is a continuous one-parameter family of SICs in dimension $3$ \cite{Appleby13}. For any of them, we may construct the SIC Q-reps. In Zhu's paper we easily see that ceiling negativity is insensitive to the SIC chosen---all that matters is that the defining property of a SIC is satisfied. This turns out to also be true for sum negativity\footnote{We won't explore it further here, but the situation is more interesting. Although the sum negativity is insensitive to the value of the parameter $t$ which defines the inequivalent SIC Q-reps, the quantum states whose quasiprobability representations achieve these sum negativity values \textit{do} depend on the parameter. It turns out that the sum negativity for $\{Q^-_j\}$ constructed from the Hesse SIC ($t=0$ in \cite{Appleby13}) is achieved by a complete set of mutually unbiased bases \cite{Fields89}, that is, $12$ ($=d(d+1)$) vectors which form four orthogonal bases such that any vector from one basis has an equal overlap with any vector from another basis. For all the other SICs in dimension $3$, the states which achieve the sum negativity of $\{Q^-\}$ form a single basis instead. The complete set of mutually unbiased bases also turns out to be the set of states which minimize the Shannon entropy in the Hesse SIC representation \cite{Szymusiak14,Stacey16}.\label{sicsandmubs}}. What's more, the sum negativities of $\{Q^+_j\}$ and $\{Q_j^-\}$ are equal:
\begin{equation}
    N^1(\{Q_j^\pm\})=\frac{1}{3}.
\end{equation}
Recall that Theorem \ref{zhu} establishes that the upper and lower bounds for ceiling negativity over all Q-reps are achieved by the SIC Q-reps. We see now that in dimension $3$, the sum negativities of $\{Q^+_j\}$ and $\{Q_j^-\}$ are equal. This tells us that if any $d=3$ Q-rep has a sum negativity other than $1/3$, the analog of Theorem \ref{zhu} does not hold for sum negativity in dimension 3.

For $d=4$,
 \begin{equation}\label{niced4}
     N^1(\{Q_j^+\})=\frac{1}{2},
 \end{equation}
 and
 \begin{equation}\label{crazyd4}
 \begin{split}
 N^1(\{Q_j^-\})&=-\frac{1}{16}\left(5+\sqrt{5}-2\sqrt{2\left(1+\sqrt{5}\right)}-2\sqrt{23-2\sqrt{5}+2\sqrt{-22+10\sqrt{5}}}\right)\\
 &\approx 0.420967.
\end{split}
\end{equation}
Value \eqref{niced4} is surprisingly nice. We will comment briefly on this in Section \ref{further}.  Value \eqref{crazyd4} is shocking. However, in light of recent results relating the SIC problem to algebraic number theory (\cite{RCF-SIC}, see the contributions of Appleby \textit{et. al} \cite{Appleby17} and Bengtsson \cite{Bengtsson16} to this volume for a review), it is worth mentioning a few possibly relevant facts about this number and how it arose.

The sum negativity for the $d=4$ SIC Q-rep is ($1/4$ times) the minimal eigenvalue of certain $7$-element partial sums of the SIC Q-rep matrices. The characteristic polynomial which has this eigenvalue as a root is:
\begin{equation}\label{charpoly}
    -\frac{1293}{32} + \frac{293\sqrt{5}}{32} - \sqrt{5 \left(22 + 29 \sqrt{5}\right)} + \left(\frac{129}{8} - \frac{35 \sqrt{5}}{8} + \sqrt{2\left(31 + 17 \sqrt{5}\right)}\right) x + \frac{21}{2} x^2 - 7 x^3 + x^4.
\end{equation}
The coefficients in equation \eqref{charpoly} are members of the field $\mathbb{Q}\left(\sqrt{5} +\sqrt{362 + 313 \sqrt{5}}\right)$. The factor multiplying $-1/16$ in \eqref{crazyd4} is an algebraic integer, but not an algebraic unit. Finally, the minimal polynomial for \eqref{crazyd4} is degree $8$.

For $d=5$, $N^1(\{Q^+\})\approx 0.584277$ and $N^1(\{Q^-\})\approx 0.501957$. These answers are numerically correct, but do not lend themselves readily to conversion to exact values.

\section{Weyl-Heisenberg Q-reps in $d=3$}\label{whnqpr}
We say that a set of vectors is \textit{group covariant} if it is the orbit of some group action on an initial vector, which we call the \textit{fiducial}. All known SICs are group covariant and all but one of those are covariant with respect to the \textit{Weyl-Heisenberg (WH) group}. In dimension $d$, let $\omega_d=e^{2\pi i/d}$ be a $d$th root of unity, and define the shift and phase operators
\begin{equation}
    X\ket{j}=\ket{j+1} \text{, } Z\ket{j}=\omega_d^j\ket{j},
\end{equation}
where the shift is modulo $d$. Products of powers of $X$ and $Z$ and powers of $\omega_d$ define the WH group. The order of the WH group in dimension $d$ is $d^3$, but for the purposes of constructing a measurement operator or Q-rep, we can neglect the phase factors. In other words, for some fiducial matrix $Q_0$, the other matrices in the orbit would take the form
\begin{equation}
    Q_{ij}=X^iZ^jQ_0(X^iZ^j)^\dag.
\end{equation}
The only known exception to WH covariance for SICs is the Hoggar SIC in dimension 8, but even this outlier is group covariant with respect to the tensor product of three $d=2$ WH groups. What about Q-reps? By their construction, the non-Hoggar SIC Q-reps are WH covariant and these Q-reps achieve the bounds for ceiling negativity in all dimensions (provided a SIC exists in that dimension). Thus, if we were to consider any subset of the full space of Q-reps for computational study, the set of
WH covariant Q-reps (WH Q-reps) is likely the best starting point. In any case, due to the ubiquity of the WH group in quantum information theory, the bounds of sum negativity within WH Q-reps may be of independent interest.

A Q-rep is associated with an orthogonal basis of operators $\{Q_j\}$ with norm $3$, that is,
\begin{equation}\label{ortho3}
    \text{Tr}(Q_iQ_j)=3 \delta_{ij}.
\end{equation}
Therefore, the general conditions for a WH Q-rep may be obtained by requiring that the elements of the WH orbit of a general unit-trace Hermitian matrix
\begin{equation}
\begin{bmatrix}\label{complexmat}
z & y& x\\
y^*&w&v\\
x^*&v^*&1-z-w
\end{bmatrix}
\end{equation}
satisfy equation \eqref{ortho3}. Imposing this condition results in a number of equations which can be algebraically simplified to the following three:

\begin{equation}\label{complexgenconds}
\begin{split}
z^2+zw+w^2&=z+w,\\
|y|^2+|x|^2+|v|^2&=1, \\
xy+y^*v+v^*x^*&=0.
\end{split}
\end{equation}
In terms of real variables, an arbitrary unit-trace Hermitian matrix
\begin{equation}
\begin{bmatrix}
a & b+ic& d+ie\\
b-ic&f&g+ih\\
d-ie&g-ih&1-a-f
\end{bmatrix}
\end{equation}
is a WH Q-rep fiducial if
\begin{equation}\label{realgenconds}
\begin{split}
a^2+af+f^2&=a+f,\\
b^2+c^2+d^2+e^2+g^2+h^2&=1, \\
dg+bd+bg+ch&=ec+eh, \\
cd+be+bh&=cg+dh+eg.
\end{split}
\end{equation}

From \eqref{complexgenconds} or \eqref{realgenconds} we can see that $d=3$ WH Q-rep fiducials define a $4$ dimensional subspace of the $8$ dimensional space of $3\times3$ unit-trace Hermitian matrices. The main diagonal elements are independent of the off-diagonal elements and satisfy the equation of an ellipse (the first equation in \eqref{complexgenconds} or \eqref{realgenconds}). The magnitude of the off-diagonal elements in \eqref{complexmat} lie on the unit 2-sphere, but their exact values only lie at points where the expression $xy+y^*v+v^*x^*$ vanishes.

The following two matrices are explicit examples of $d=3$ \textit{Non-SIC} WH Q-rep fiducials:
\begin{equation}\label{minsumneg}
    Q^{\text{min}}=
\begin{bmatrix}
    0 & -\frac{1}{3}+\frac{i}{3} & \frac{2}{3}-\frac{i}{3} \\
   -\frac{1}{3}-\frac{i}{3}  &1&-\frac{1}{3}+\frac{i}{3}\\
   \frac{2}{3}+\frac{i}{3} & -\frac{1}{3}-\frac{i}{3}&0
\end{bmatrix},
\end{equation}
\begin{equation}\label{maxsumneg}
    Q^{\text{max}}=
\begin{bmatrix}
    0 & \left(-\frac{1}{3}+\frac{\sqrt{7}}{12}\right)+\frac{i}{4} & \left(\frac{2}{3}+\frac{\sqrt{7}}{12}\right)-\frac{i}{4}\\
   \left(-\frac{1}{3}+\frac{\sqrt{7}}{12}\right)-\frac{i}{4} &1&\left(-\frac{1}{3}+\frac{\sqrt{7}}{12}\right)+\frac{i}{4}\\
  \left(\frac{2}{3}+\frac{\sqrt{7}}{12}\right)+\frac{i}{4}& \left(-\frac{1}{3}+\frac{\sqrt{7}}{12}\right)-\frac{i}{4} &0
\end{bmatrix}
.
\end{equation}
In the following section we will see why they are presented with the designations ``max'' and ``min''.

\section{Sum Negativity Bounds in $d=3$}\label{sumnegd3}
The WH Q-reps generated by \eqref{minsumneg} and \eqref{maxsumneg} will be denoted $\{Q^\text{min}_j\}$ and $\{Q^\text{max}_j\}$ respectively. Their sum negativities are
\begin{equation}\label{minval}
    N^1(\{Q^\text{min}_j\})=\frac{1}{3}\left(2\cos\frac{\pi}{9}-1\right)\approx 0.293128,
\end{equation}
and
\begin{equation}\label{maxval}
N^1(\{Q^\text{max}_j\})=\frac{2}{9}\left(\sqrt{7}-1\right)\approx 0.365723.
\end{equation}
Recall from Section \ref{sumnegsic} that the sum negativities of both SIC Q-reps is $1/3$, so $\{Q^\text{min}_j\}$ and $\{Q^\text{max}_j\}$ are explicit counterexamples to Zhu's theorem for sum negativity. It turns out that \eqref{maxval} is a strict upper bound on the sum negativity, not only of WH Q-reps, but of all Q-reps in dimension 3. We start with two lemmas.
\begin{lemma}\label{stationary}
    For quasiprobability vectors with $d^2$ elements lying in the sphere of radius $\sqrt{1/d}$, the stationary points for sum negativity are:
\begin{enumerate}
    \item Those vectors whose entries consist only of two distinct values.
    \item Those vectors whose entries consist only of three distinct values including zero.
\end{enumerate}
\end{lemma}
\begin{proof}

Quasiprobability vectors in the sphere of radius $\sqrt{1/d}$ satisfy the following constraints:
\begin{equation}\label{constraints}
\sum_{j=1}^{d^2}\mathfrak{p}(j)=1 \quad \text{and} \quad \sum_{j=1}^{d^2}\mathfrak{p}(j)^2=\frac{1}{d}.
\end{equation}
The definition of the sum negativity of a quasiprobability vector \eqref{sumneg} gives us
\begin{equation}\label{sumneg}
    N^1(\mathfrak{p})=|\!|\mathfrak{p}^{(-)}|\!|_1=\sum_{j=1}^{d^2}\frac{|\mathfrak{p}(j)|-\mathfrak{p}(j)}{2}=\frac{1}{2}\sum_{j=1}^{d^2}|\mathfrak{p}(j)|-\frac{1}{2}.
\end{equation}
From this it is clear that the stationary points of the sum negativity are exactly the stationary points of the sum of the absolute values. Absolute values are often difficult to deal with in optimization problems, but it turns out that the function $|x|$ may be approximated efficiently by $\sqrt{x^2+c}$ where $c$ is taken to zero after any differentiation \cite{Argaez14}. Thus, we want to extremize the function
\begin{equation}\label{sumnegdiff}
\sum_{j=1}^{d^2}\sqrt{\mathfrak{p}(j)^2+c}
\end{equation}
subject to constraints \eqref{constraints} in the small $c$ limit. To do this we construct a Lagrangian
\begin{equation}
    \mathcal{L}(\mathfrak{p},\lambda,\mu)=\sum_{j=1}^{d^2}\left(\sqrt{\mathfrak{p}(j)^2+c}\right)-\lambda\left(\sum_{j=1}^{d^2}\mathfrak{p}(j)-1\right)-\mu\left(\sum_{j=1}^{d^2}\mathfrak{p}(j)^2-\frac{1}{d}\right)
\end{equation}
where $\lambda$ and $\mu$ are Lagrange multipliers. Varying this Lagrangian, we see that the stationary points must satisfy
\begin{equation}
\frac{\mathfrak{p}(j)}{\sqrt{\mathfrak{p}(j)^2+c}}-\lambda-2\mu \mathfrak{p}(j)=0
\end{equation}
for all $j$. This expression has $4$ solutions, two of which become zero in the $c=0$ limit:
\begin{equation}\label{entries}
        \mathfrak{p}(j)=\frac{\pm 1-\lambda}{2\mu}\quad\text{or}\quad \mathfrak{p}(j)=0.
\end{equation}
    Consider the case where $\mathfrak{p}(j)\neq 0$ for all $j$. Now, in order for the constraints \eqref{constraints} to hold, some number $n$ of the entries are $\frac{-1-\lambda}{2\mu}$ and the other $d^2-n$ are $\frac{1-\lambda}{2\mu}$. That is,
\begin{equation}\label{solveformultipliers}
    n\left(\frac{-1-\lambda}{2\mu}\right)+(d^2-n)\left(\frac{1-\lambda}{2\mu}\right)=1 \quad \text{and} \quad n\left(\frac{-1-\lambda}{2\mu}\right)^2+(d^2-n)\left(\frac{1-\lambda}{2\mu}\right)^2=\frac{1}{d}.
\end{equation}
It is easy to verify that $0<n<d^2$. In this case the quasiprobability vector $\mathfrak{p}$ consists of two distinct values. If $\mathfrak{p}(j)=0$ for $m$ of the indices then $d^2-n-m$ of the entries are $\frac{1-\lambda}{2\mu}$ and the appropriately modified form of \eqref{solveformultipliers} holds in which case the quasiprobability vector $\mathfrak{p}$ consists of three distinct values including zero.
\end{proof}
A stationary point can be a local maximum, a local minimum, or a saddle point. Sufficient conditions for maxima and minima in Lagrangian systems with equality constraints are known. We will need the following tool (which can be found in chapter 2 of \cite{Pike86}):

\begin{thm}[Sufficient Conditions for Constrained Maxima]
    Consider a constrained maximization problem for a twice-differentiable function of $n$ variables $y(\mathbf{x})$ with $m$ twice-differentiable equality constraints $f_i(\mathbf{x})=0$, $i=1,\ldots,m$. The Lagrangian is
    \begin{equation}\label{lagrangian}
\mathcal{L}(\mathbf{x},\boldsymbol{\lambda})=y(\mathbf{x})-\sum_{i=1}^m \lambda_i f_i(\mathbf{x}),
\end{equation}
    where $\lambda_i$ are Lagrange multipliers. If there exist vectors $\mathbf{x}^*$ and $\boldsymbol\lambda^*$ such that
$\partial_i\mathcal{L}(\mathbf{x}^*,\boldsymbol{\lambda}^*)=0$, $i=1, \ldots,n$ and $f_i(\mathbf{x}^*)=0$, $i=1,\ldots,m$ and if
\begin{equation}\label{bddhess}
(-1)^s
\begin{vmatrix}
    \partial_{11}\mathcal{L} & \cdots & \partial_{1s}\mathcal{L} & \partial_1f_{1} & \cdots & \partial_1f_{m}\\
\vdots & & \vdots& \vdots & & \vdots \\
    \partial_{1s}\mathcal{L} & \cdots & \partial_{ss}\mathcal{L} & \partial_sf_{1}& \cdots & \partial_sf_{m}\\
\partial_1f_{1} & \cdots & \partial_sf_{1} & 0 & \cdots & 0\\
\vdots & & \vdots & \vdots & & \vdots \\
\partial_1f_{m} & \cdots & \partial_sf_{m}& 0 &\cdots & 0
\end{vmatrix}
>0
\end{equation}
    for $s=m+1,\ldots, n$ ($\partial_{ij}\mathcal{L}$ indicates a second partial derivative of $\mathcal{L}(\mathbf{x},\boldsymbol{\lambda})$ with respect to $x_i$ and $x_j$ evaluated at $\mathbf{x}^*$ and $\boldsymbol{\lambda}^*$ and $\partial_if_{j}$ indicates the first partial derivative of constraint function $f_j(\mathbf{x})$ with respect to $x_i$ evaluated at $\mathbf{x}^*$), then $y(\mathbf{x})$ has a strict local maximum at $\mathbf{x}^*$.
\end{thm}
A consequence of this is the following:
\begin{lemma}\label{extremelemma}
The stationary points for sum negativity of quasiprobability vectors with $d^2$ elements lying in the sphere of radius $\sqrt{1/d}$ are all local maxima or global minima.
\end{lemma}
\begin{proof}
    From Lemma \ref{stationary} we know that the stationary points for sum negativity of quasiprobabilities lying on the sphere of radius $\sqrt{1/d}$ are those consisting of exactly two distinct elements or those consisting of three if one of them is zero. If there are two distinct elements, they can either both be nonnegative or of opposite signs. If they are both nonnegative, the sum negativity is zero which is the global minimum sum negativity value. If there are three distinct elements and
    they are all nonnegative, then it is also a global minimum for sum negativity. We will now show that when the nonzero elements of a stationary point are of opposite signs, it is a local maximum.

    Consider first the stationary points with two distinct elements. In terms of the Lagrange multipliers from the proof of Lemma \ref{stationary}, the stationary quasiprobability vectors are comprised of opposite signed values when $\mu\neq0$ and $-1<\lambda<1$. Without loss of generality we choose $\frac{-1-\lambda}{2\mu}$ to be the negative value so that we may enumerate the number of negative entries with index $n$ as in \eqref{solveformultipliers}. This amounts to the further restriction $\mu>0$. Note that by substituting $\frac{-1-\lambda}{2\mu}=a$ and
    $\frac{1-\lambda}{2\mu}=b$ in \eqref{solveformultipliers}, we can solve for the positive and negative entries in the quasiprobability vector in terms of the dimension $d$ and number of negative entries $1\leq n<d(d-1)$:
\begin{equation}\label{aval}
a=\frac{1}{d^2}+\frac{1}{d^2}\sqrt{\frac{n(d-1)}{d^2-n}}-\sqrt{\frac{d-1}{n(d^2-n)}}<0,
\end{equation}
\begin{equation}\label{bval}
b=\frac{1}{d^2}+\frac{1}{d^2}\sqrt{\frac{n(d-1)}{d^2-n}}>0.
\end{equation}

    Sum negativity is invariant to the ordering of the entries in a quasiprobability vector so we need only prove that one ordering is a local maximum for each $n$. Without loss of generality, we fix $\mathfrak{p}(1)=a$ and demand that $\mathfrak{p}(2)$ and $\mathfrak{p}(3)$ are not both also equal to $a$. Consider again the Lagrangian \eqref{lagrangian}. Taking the relevant derivatives and $c$ to zero, we may construct the bordered Hessian matrix from \eqref{bddhess} for our problem.
    Therefore, if $\vec{\mathfrak{p}}$ is a stationary point and if
\begin{equation}
(-1)^s
\begin{vmatrix}\label{sufficient}
-2\mu  & \cdots & 0 & 1 & 2\mathfrak{p}(1)\\
\vdots & \ddots & \vdots & \vdots&\vdots\\
0 & \cdots & -2\mu &1&2\mathfrak{p}(s) \\
1 & \cdots & 1 & 0 & 0\\
2\mathfrak{p}(1) & \cdots & 2\mathfrak{p}(s) & 0 & 0
\end{vmatrix}
>0
\end{equation}
    for $s=3,\ldots,d^2$, then $\vec{\mathfrak{p}}$ is a local maximum in sum negativity.

Recall the following identity for matrix blocks $A$, $B$, $C$, and $D$ which holds when $A$ is invertible:
\begin{equation}\label{detid}
\begin{vmatrix}
A & B\\
C& D
\end{vmatrix}
=|A||D-CA^{-1}B|.
\end{equation}
Using \eqref{detid}, the determinant of the bounded Hessian matrix in \eqref{sufficient} is
\begin{equation}
(-1)^s(2\mu)^{s-1}\left(4s\sum_{i=1}^s\mathfrak{p}(i)^2-4\sum_{i,j=1}^s \mathfrak{p}(i)\mathfrak{p}(j)\right).
\end{equation}
With the additional $(-1)^s$ term in equation \eqref{sufficient}, we see that the inequality is satisfied if
\begin{equation}\label{suffinequal}
s\sum_{i=1}^s\mathfrak{p}(i)^2-\sum_{i,j=1}^s \mathfrak{p}(i)\mathfrak{p}(j)>0
\end{equation}
for all $s$. Let $l$ denote the number of negative elements in the truncated quasiprobability vectors appearing in \eqref{sufficient}. Note that $1\leq l<s$ by our earlier assumption about the first three elements. Then
\begin{equation}
s\sum_{i=1}^s\mathfrak{p}(i)^2-\sum_{i,j=1}^s \mathfrak{p}(i)\mathfrak{p}(j)=s(l a^2+(s-l)b^2)-(la+(s-l)b)^2=(a-b)^2(s-l)l>0.
\end{equation}
Thus the stationary points consisting of two distinct values, one negative and one positive, are local maxima.

The values of $a$ and $b$ in the case with three distinct values are more complicated, but may still be obtained. Otherwise the proof in this case carries through in the same way as above with a slight modification in the last step. Let $m$ denote the number of elements equal to zero in the truncated quasiprobability vectors appearing in \eqref{sufficient}. Note that $l+m\leq s$. Then it can be verified that
\begin{equation}
s\sum_{i=1}^s\mathfrak{p}(i)^2-\sum_{i,j=1}^s \mathfrak{p}(i)\mathfrak{p}(j)=s(l a^2+(s-l-m)b^2)-(la+(s-l-m)b)^2>0,
\end{equation}
which completes the proof.
\end{proof}
We may now return to the question of extremality among Q-reps. We say that a pure state $\rho$ \textit{achieves the sum negativity} if it is an eigenvector with eigenvalue magnitude equal to $d$ times $N^1(\{Q_j\})$ of one of the partial sum matrices of the Q-rep $\{Q_j\}$.
\begin{thm}\label{extremenqpr}
    If the quasiprobability representation of a state which achieves the sum negativity of a Q-rep $\{Q_j\}$ consists of two distinct elements or three including zero, then $N^1(\{Q_j\})$ is a local maximum among all Q-reps.
\end{thm}
\begin{proof}
    Recall that the quasiprobability representation of a pure state with respect to a Q-rep lies in the sphere of radius $\sqrt{1/d}$. Also recall that the convexity of the sum negativity function implies that any state which achieves the sum negativity for a Q-rep $\{Q_j\}$ is a pure state. Therefore, if, as we vary $\{Q_j\}$ over the space of Q-reps, the quasiprobability representation of a pure state which achieves the sum negativity of $\{Q_j\}$ consists of two distinct values or three
    distinct values including zero for some Q-rep $\{Q_j'\}$, then by Lemmas \ref{stationary} and \ref{extremelemma}, there is a local maximum or global minimum of the sum negativity function at $\{Q_j'\}$. As we know the appearance of negativity is inevitable, the global minimum possibility is avoided.
\end{proof}
\begin{thm}\label{exactupperbound}
 The exact upper bound for sum negativity among Q-reps in dimension $3$ is $\frac{2}{9}(\sqrt{7}-1)$.
\end{thm}
\begin{proof}
    When $d=3$, explicit calculation reveals that the sum negativity of a quasiprobability vector lying in the sphere of radius $\sqrt{1/d}$ constructed with $n$ values equal to $a<0$, $m$ values equal to $0$, and $9-n-m$ values equal to $b>0$ is maximized when $n=2$ and $m=0$. Those values are, from \eqref{aval} and \eqref{bval},
\begin{equation}
a=\frac{1}{9}(1-\sqrt{7})
\text{ and }
b=\frac{1}{63}(7+2\sqrt{7}).
\end{equation}
    Lemmas \ref{stationary} and \ref{extremelemma} and the fact that our domain has no boundary (lying on the sphere of radius $\sqrt{1/d}$ and quasiprobability normalization together define a $(d^2-2)$--sphere) imply that the sum negativity of this quasiprobability vector is the global maximum value over this domain. The quasiprobability representations of the states which achieve \eqref{maxval} with respect to the Q-rep $\{Q_j^\text{max}\}$ consist of these values, and so, by Theorem
    \ref{extremenqpr}, $\frac{2}{9}(\sqrt{7}-1)$ is the maximum value
    for sum negativity over all Q-reps in dimension $3$.
\end{proof}

It is important to note that while Theorem \ref{exactupperbound} shows that \eqref{maxval} is the strict upper bound for sum negativity among all Q-reps, it does not imply that $\{Q_j^\text{max}\}$ is the unique Q-rep which achieves this bound. Numerical searching suggests that, for the fiducial main diagonal $\{0,1,0\}$, the WH orbit of \eqref{maxsumneg} is the unique Q-rep which achieves the sum negativity \eqref{maxval}, but that there is at least one other fiducial main diagonal which achieves this bound, namely the main diagonal corresponding to the major axis vertices of the ellipse defined by the first equation in \eqref{complexgenconds}: $\{\frac{1}{3},\frac{1}{3}+\frac{1}{\sqrt{3}}, \frac{1}{3}-\frac{1}{\sqrt{3}}\}$. Unfortunately, for this fiducial, we were unable to convert the numerical result to exact values.

What can be said about the lower bound for sum negativity in dimension $3$? So far, less is known, but we present the following numerically motivated conjecture:
\begin{conj}\label{lowerboundconj}
    The exact lower bound for sum negativity among WH Q-reps in dimension $3$ is $\frac{2}{3}(\cos\frac{\pi}{9}-\frac{1}{2})$.
\end{conj}
\noindent This statement resisted our attempts to prove it in a fashion similar to Theorem \ref{exactupperbound} because for every Q-rep there exist states with zero negative elements in their quasiprobability representation (for example, the maximally mixed state). The Q-rep vector corresponding to the minimal eigenstate over all partial sums of $\{Q_j^\text{min}\}$ consists of three distinct values which each appear three times:
\begin{equation}\label{minqprobvals}
\frac{1}{9}\left(1-2\cos\frac{\pi}{9}\right) \text{, } \frac{1}{9}\left(1+2\cos\frac{2\pi}{9}\right) \text{, and } \frac{1}{9}\left(1+2\cos\frac{\pi}{9}-2\cos\frac{2\pi}{9}\right).
\end{equation}
The values in \eqref{minqprobvals} may not have significance as deep as those in the upper bound quasiprobability vector, however, as numerical searching has revealed that \eqref{minval} is achieved by a WH Q-rep for \textit{every} valid main diagonal. The off-diagonal elements in other cases were too difficult to convert to exact values. The fact that there seems to be a WH Q-rep which achieves \eqref{minval} for any main diagonal satisfying the first equation in \eqref{complexgenconds} suggests
the lower bound among Q-reps may not be saturated by WH Q-reps. Additionally, the extremal properties of general quasiprobability vectors lying on the sphere of radius $\sqrt{1/d}$ do not come to our aid when we try to find the lower bound because there is no obvious reason to hope that the process of maximizing over quantum state space (in the definition of a Q-rep sum negativity) followed by minimizing over all Q-reps (to find the lower bound for sum negativity) should result in one of the
local maxima for general quasiprobability vectors on the sphere of radius $\sqrt{1/d}$ (recall, of course, that it cannot result in the global minimum of zero negativity on this sphere because we know that the appearance of negativity is inevitable). Before, we were maximizing over both; in some sense we got lucky that the global maximum sum negativity for a quasiprobability vector on this sphere was achieved by a Q-rep vector.

\section{Further Observations about SIC Q-reps}\label{further}
In Section \ref{sumnegsic}, we noted the appearance of a rational value for the sum negativity of $\{Q_j^+\}$ in dimension $4$.  The quasiprobability vectors which achieve the sum negativity of $\{Q_j^+\}$ are of a special and familiar form; they consist of only the values $-1/8$ and $1/8$. We know from Theorem \ref{extremenqpr}, therefore, that this Q-rep is a local maximum for sum negativity among Q-reps in dimension $4$. In fact, following the exact same procedure as
in Theorem \ref{exactupperbound}, we see that $N^1(\{Q_j^+\})=1/2$ is the exact upper bound over all Q-reps!

In footnote \ref{sicsandmubs}, we mentioned that the states which achieve the sum negativity of $\{Q^-_j\}$ for the Hesse SIC in dimension $3$ form a complete set of mutually unbiased bases. The states which achieve the sum negativity for $\{Q^-_j\}$ in dimension $4$ also form a structure of possible interest. They consist of a set of $16$ vectors which have two nontrivial squared overlaps. In the terminology of reference \cite{Zauner99}, these states form a quantum design of degree $2$.

Although dimension $5$ was the last in which we were able to explicitly calculate the sum negativity for the SIC Q-reps by exhaustive combinatorial searching, we suspect that we have found the correct sum negativity for $\{Q_j^-\}$ constructed with the Hoggar SIC in dimension $8$. Rather than calculating the eigenvalues of every partial sum matrix (since this is infeasible for $2^{64}$, $8\times8$ matrices), we used a numerical local maximization procedure and around $10^6$ random pure state
seeds. The overall maximum value we found, $7/8$, occurred frequently in our data and is significantly larger than all of the smaller local maxima. Of course, we could still be falling short of the global maximum value if it occurs at very hard to access positions. The states whose quasiprobability representations achieve the sum negativity of $7/8$ consist of $28$ copies of value $-1/32$ and $36$ copies of value $5/96$. Therefore, by Theorem \ref{extremenqpr}, if these states achieve the actual
sum negativity, then $\{Q_j^-\}$ constructed with the Hoggar SIC is a local maxima among all Q-reps in dimension 8. Surprisingly, these quantum states also minimize the Shannon entropy of their Hoggar SIC representations and thus compose the ``twin'' Hoggar SIC \cite{Szymusiak2015,Stacey16}. This result and the one for $\{Q^-_j\}$ in dimension $4$ in the previous paragraph parallel the one mentioned in footnote \ref{sicsandmubs}.

\section{Discussion}\label{discussion}
In QBism, quantum states are probability distributions over a set of possible outcomes for an appropriately chosen measurement. With the understanding that quantum theory is an addition to coherence which rational agents should use to help inform their expectations for future experiences in terms of their past ones, QBists hope that the structure of quantum theory, and in particular its boundaries, can be made to suggest nature's motives. A century of quantum foundational debate should by now
have convinced us that these motives will not conform to our prejudices about reality. Although the weirdness of quantum theory has convinced some that we can no longer pretend physics is more than an exercise in instrumentalism, QBists are optimistic that there are ubiquitous and recognizably physical statements about nature yet to be made.

Prior to Zhu's paper, we had focused on characterizing the bounds of quantum theory only from within probability theory. His results reveal this approach was nearsighted. In the initial stages of this project, we had hoped to find further evidence for the centrality of SICs in quantum theory by showing that Zhu's theorem extends to another natural measure of negativity for quasiprobability representations. Indeed, it is interesting to find that the SICs do not generally play the same role in this alternate context. Why do they not? Do they still always play
\textit{some} role which is not immediately apparent? Maybe there are other families of Q-reps constructed from SICs or another structure which naturally play the same role for sum negativity. $\{Q_j^+\}$ in dimension 4 did achieve the upper bound for sum negativity. Does this happen again? Although we cannot calculate it exactly, we have some numerical evidence that the sum negativities for the Hoggar SIC Q-reps differ from the non-Hoggar SIC Q-reps in dimension $8$. This suggests that there might be an essential
relation between sum negativity and group covariance. Furthermore, the appearance of the complete set of mutually unbiased bases in dimension $3$ (mentioned in footnote \ref{sicsandmubs}) and the ``twin'' Hoggar SIC in dimension $8$ suggest a deep connection between sum negativity and minimizing Shannon entropy which warrants further exploration.

Due to the connection to contextuality, the bounds on sum negativity and $N^p$ negativity in general are likely to be of interest to the quantum computation community. A natural further direction for this research is the consideration of negativities other than ceiling and sum negativity. Perhaps the next to consider is the only non-convex integral negativity, $N^0$, which tells us the maximum number of negative elements which can appear in a Q-rep vector. On the other hand, $N^2$ negativity makes use of the most familiar distance function, and, as such, may warrant special attention. In Section \ref{whnqpr} we
established the general conditions for a $d=3$ WH Q-rep. This result and any analogous results in higher dimensions\footnote{And in dimension $8$, it may be interesting to look at general $\text{WH}\otimes\text{WH}\otimes\text{WH}$ covariant Q-reps.} may be of independent interest. Likewise, in order to pursue the exact lower bound for sum negativity in dimension $3$, we need strategies to construct non-WH Q-reps. Towards this, the general structure and symmetries inherent in Q-reps warrants exploration. It may further be interesting
to consider what can be said about quasiprobability representations obtained from non-orthogonal bases or even redundant operator frames in the negativity context.
\section*{Acknowledgements}
The authors would like to thank Marcus Appleby, Blake Stacey, and Huangjun Zhu for helpful discussions and suggestions. This research was supported by the Foundational Questions Institute Fund on the Physics of the Observer (grant FQXi-RFP-1612), a donor advised fund at the Silicon Valley Community Foundation.

\bibliographystyle{utphys}

\bibliography{master}

\end{document}